\RequirePackage{fix-cm}
\documentclass[smallextended]{svjour3}       % onecolumn (second format)
\smartqed  % flush right qed marks, e.g. at end of proof
\usepackage{graphicx}
\usepackage{natbib}
\usepackage{graphicx}

\usepackage{amsmath}
\usepackage{amsfonts}
\usepackage{amssymb}
\usepackage{graphicx}
\usepackage{tikz}

\usepackage[english]{babel}
 
\usepackage{amsthm}

\usepackage{arydshln}

\usepackage{tikz,fullpage}
\usetikzlibrary{arrows,%
                petri,%
                topaths}%
\usepackage{tkz-berge}
\usepackage[position=top]{subfig}
\usepackage{pgf}

\usetikzlibrary{arrows,automata}
\usepackage{inputenc}

\theoremstyle{definition}

\theoremstyle{remark}
\newtheorem*{observation}{Observation}
\begin{document}

\title{Phylosymmetric algebras: mathematical properties of a new tool in phylogenetics }

\author{Michael Hendriksen         \and
        Julia A. Shore 
}

\institute{M. Hendriksen \at
              Centre for Research in Mathematics and Data Science, Western Sydney University, NSW, Australia \& Institut f{\"u}r Molekular Evolution, Heinrich-Heine Universit{\"a}t\\
              \email{michael.hendriksen@hhu.de}             \\
             \emph{Substantial parts of MH's research were carried out at both WSU and HHU}
             \and
           J. Shore \at
              University of Tasmania, Churchill Avenue, Sandy Bay, Tasmania, Australia 7005
}

\date{Received: date / Accepted: date}

\maketitle

\begin{abstract}
In phylogenetics it is of interest for rate matrix sets to satisfy closure under matrix multiplication as this makes finding the set of corresponding transition matrices possible without having to compute matrix exponentials. It is also advantageous to have a small number of free parameters as this, in applications, will result in a reduction of computation time. We explore a method of building a rate matrix set from a rooted tree structure by assigning rates to internal tree nodes and states to the leaves, then defining the rate of change between two states as the rate assigned to the most recent common ancestor of those two states. We investigate the properties of these matrix sets from both a linear algebra and a graph theory perspective and show that any rate matrix set generated this way is closed under matrix multiplication. The consequences of setting two rates assigned to internal tree nodes to be equal are then considered. This methodology could be used to develop parameterised models of amino acid substitution which have a small number of parameters but convey biological meaning.
\keywords{Phylogenetic methods \and graph theory \and matrix algebras \and rate matrices \and matrix models \and rooted trees}
\end{abstract}

\section*{Acknowledgements}
MH thanks the Volkswagen Foundation 93\_046 grant for support during research at HHU and the Australian Postgraduate Award for support during research at WSU.

\section{Introduction}
\label{intro}
Phylogenetics is the study of constructing phylogenetic trees that represent evolutionary history. Analysis of RNA, DNA and protein sequence data with the use of continuous time Markov chains to measure the frequency of occurrence of point mutations is commonly employed in this field. From a continuous time Markov chain, transitions matrices (whose matrix entries represent probabilities of a change of state for a set time period) and rate matrices (whose entries represent the rates of change between states) can be generated. Transition matrices in phylogenetics are typically classified as either empirical, where the transition probabilities are values which have been calculated by analysing sequence data, or parameterised, where transition probabilities are represented by free parameters which are chosen to fit data as needed \citep{yang2014molecular}. Given that a parameterised transition matrix contains free parameters, it can be thought of as a set of transition matrices and such a set is often referred to as a model where the set of transition matrices is denoted by $ \mathcal{M} $ and the set of corresponding rate matrices is denoted by $\mathcal{Q}$. Parameterised models are often developed to be consistent with biological and chemical mechanisms (e.g. the K2P model \cite{kimura1980simple} captures the fact that it is chemically easier to substitute a purine for a purine or a pyrimidine for a pyrimidine) but sometimes they are developed to satisfy mathematical properties. Some parameterised models are more complicated than setting two rates to be equal to each other e.g. there are multiplicative constraints on matrix entries. In this paper, however, we will only be looking at models whose constraints are that some rates are equal to other rates.

The Lie Markov models (LMM) \citep{sumner2012lie, fernandez2015lie} are a set of parameterised DNA rate substitution models. Their construction is based on mathematical properties of matrices: each rate matrix model in this set forms a Lie algebra (note that a Lie algebra in this context can be defined as a matrix vector space which is closed under the operation $[A,B] = AB-BA$) as this guarantees that each transition matrix set is closed under matrix multiplication. In a study following this, \citet{shore2015lie} found that if a rate matrix set, $\mathcal{Q}$, forms a matrix algebra (a matrix algebra we define as a matrix vector space which is closed under matrix multiplication, any matrix algebra is automatically a Lie algebra), the set of corresponding transition matrices is $ \{ I + Q: Q \in \mathcal{Q}, $ det$(I+Q) \not=0 \} $. This makes finding the space of corresponding transition matrices a straightforward process compared to the usual practice of having to calculate matrix exponentials, which is notoriously computationally expensive \citep{moler1978nineteen}, although unfortunately this does not completely absolve the necessity of calculating matrix exponentials in practice. It is therefore advantageous for a rate matrix set to form a matrix algebra. 

The study conducted by \citet{shore2019} employed a method of generating rate matrix sets from trees by labelling leaves on a rooted tree as the states and then defining the rate of change between two states to be the rate assigned to their most recent common ancestor (note that this method is explained in more detail in Section \ref{backgroundSection}). This method was used to test if certain biological mechanisms to distinguish amino acids could have developed in a serial manner (i.e. the specificity of a mechanism increased over time) and what properties of amino acids could have effected this development. To test this, the rooted trees were used to represent the increasing specificity of amino acid selection mechanisms rather than the evolution of a group of organisms. 

Their methodology, which is now the focus of this work, was used to show that there is a link between properties of amino acids (namely their polarity and the class in which their corresponding aaRS fall into) and the observed rates of change between amino acids as described in \citet{le2008improved}. Given that this methodology has already been shown to correlate with biological mechanisms, it is now proposed that it be used to develop a suite of parameterised substitution models; particularly for amino acid substitution of which the most commonly used rate matrices are empirical. The family of rate matrix sets generated by this method has previously been unexplored and we now aim to gain a mathematical understanding of these matrix sets.

In the present paper, we introduce a set of matrices associated with trees with rates associated to each interior vertex. In Section \ref{graphTsection}, we derive results on the multiplication of these matrices, and show, in the case that each rate is unique, that the matrices form a matrix algebra, which we refer to as a \textit{phylosymmetric algebra}. In Section \ref{sectionRepeatedR}, we extend this result to completely characterise all conditions for which the matrices form a matrix algebra when two rates are identical, and derive sufficient conditions for simple cases of arbitrarily many equal rates.

\section{Background}
\label{backgroundSection}

\begin{definition} 
A \textit{rooted tree} $\mathcal{T}$ on a set of taxa $X$ is a connected, directed acyclic graph with no vertices of degree-$2$ other than the root, and whose leaves (degree-$1$ vertices) are bijectively labelled by the set $X$. The vertices other than the root and the leaves are referred to as \textit{internal vertices}. Subtrees of $ \mathcal{T} $ are denoted by $T$. The set of all rooted phylogenetic trees on a set of taxa $X$ is denoted $RP(X)$.
\end{definition}

All trees in this paper are rooted trees and are permitted to be non-binary. We will henceforth refer to them as $X$-trees, or simply trees if there is no ambiguity.

If there is a directed edge from a vertex $u$ to a vertex $v$, then we say that $u$ is a \textit{parent} of $v$ and $v$ is a \textit{child} of $u$. If there is a directed path from $u$ to $v$ then $u$ is an \textit{ancestor} of $v$ and $v$ is a descendant of $u$. In particular, a parent of a vertex $v$ is always an ancestor of $v$, a child of $v$ is always a descendant of $v$, and $v$ is both an ancestor and descendant of itself. If two vertices $u$ and $v$ share a parent vertex, we say that $u$ and $v$ are \textit{siblings} of each other. 

\begin{definition}
A \textit{hierarchy} $H$ on a set $X$ is a collection of subsets of $X$ with the following properties:

\begin{enumerate}
\item $H$ contains both $X$ and all singleton sets $\{x\}$ for $x \in X$.
\item If $H_1,H_2\in H$, then $H_1 \cap H_2 = \varnothing$, $H_1 \subseteq H_2$ or $H_2 \subseteq H_1$.
\end{enumerate} 
\end{definition}

\begin{definition}
Let $\mathcal{T} \in RP(X)$ be a tree and $v$ be a vertex of $\mathcal{T}$. Then the \textit{cluster} of $\mathcal{T}$ associated with $v$ is the subset of $X$ consisting of the descendants of $v$ in $\mathcal{T}$.
\end{definition}

A collection of subsets of $X$ is a hierarchy if and only if it is the set of clusters of some rooted phylogenetic tree $\mathcal{T}$ taken over all vertices of $\mathcal{T}$ (see~\cite{steel2016phylogeny} for instance). For this reason we refer to the set of clusters of $T$ as the \textit{hierarchy} of $\mathcal{T}$, denoted $H(\mathcal{T})$.

Suppose $ \mathcal{T} $ is a tree with vertex set $V$ and leaf set $ X = \{1,2,...,n\} \subseteq V $. For each pair of vertices $ a,b $ we denote their most recent common ancestor as mrca$(a,b)$. Define a function $\omega: V \rightarrow \mathbb{R}$ that assigns a real number to each vertex of the tree. For each vertex, $u \in V$, we call $\omega(u) = \alpha$ the \textit{rate} at $u$. Define the subset $ C_{\alpha} \subseteq X\times X $ where $ (x,y) \in C_{\alpha} $ if and only if mrca$(x,y) = u$. It follows that the set $ \{ C_{\alpha}: \alpha \in V \} $ forms a partition of $ X \times X $. 

To each $ C_{\alpha} $ we associate an $ n \times n $ matrix $Q_{\alpha}$ with off diagonal entries given by

\[ (Q_{\alpha})_{xy} = 
\begin{cases} 
1 & \mbox{ if mrca}(x,y) = u, \\ 
0 & \mbox{ otherwise} \\ \end{cases}; \]
and diagonal entries
 
\[ (Q_{\alpha})_{xx} = -\# (z:(x,z) \in C_{\alpha}). \]

We refer to $Q_\alpha$ as the \textit{rate matrix associated with $\alpha$}. Note that when $u$ is a leaf on $\mathcal{T}$, the corresponding rate matrix $Q_{\alpha} = 0$, and that matrices produced by the mrca function are symmetric. The set of mrca matrices produced by a single tree form the basis for a matrix algebra (see Theorem \ref{treealgebraThm}). Therefore products in this space are symmetric, which implies that the algebra is commutative (see Lemma \ref{commutingLemma}). The intent of this paper is to investigate the properties of the resulting set of matrix algebras.

\begin{remark}
It follows quickly from the definitions that 

\[ \sum_{\alpha \in \omega(V)} Q_\alpha = J, \]
where $J$ is the $n \times n$ matrix with $1$ in each off diagonal entry and $1-n$ in each diagonal entry. In fact, if some non-leaf vertex $u$ has $m$ leaf descendants and we denote the set of all vertices that are descendants of some vertex $u$ by $V_u$, we can see that  

\[ \sum_{\alpha \in \omega(V_u)} Q_\alpha = J_u, \]
where $J_u$ is the matrix

\[ (J_u)_{ij} = 
\begin{cases} 
1 & \mbox{if $i\ne j$ and $i,j$ are descendants of $u$}, \\ 
-m & \mbox{if i=j, and} \\ 
0 & \mbox{otherwise.} \\\end{cases} \]
\end{remark}

\begin{lemma}
If the product of two symmetric matrices is also symmetric, then those two matrices commute \citep{leon2010linear}.\label{commutingLemma}
\end{lemma}

\begin{proof}
Let $A$, $B$ and $AB$ be symmetric matrices. Then we have:
\begin{align*}
    AB &= (AB)^T \\
    &= B^T A^T \\
    &= BA.
\end{align*}
\end{proof}

\begin{example}
We end this section by computing the rate matrix set associated with the tree in Figure \ref{f:RateExample}.
\begin{figure}
\centering
\begin{tikzpicture}[level 1/.style={sibling distance=14mm},level 2/.style={sibling distance=9mm},level 3/.style={sibling distance=7mm},
  every node/.style = {shape=rectangle, rounded corners,
    draw, align=center}]]
  \node {$ \alpha $}
    child[level distance = 2cm]{node{$ \beta $}
    	child[level distance = 2cm]{node{1}}
    	child[level distance = 2cm]{node{2}} }
    child[level distance = 2cm]{node{$ \gamma $}
    	child[level distance = 2cm]{node{3}}
    	child[level distance = 1cm]{node{$ \delta $}
    		child[level distance = 1cm]{node{4}}
    		child[level distance = 1cm]{node{5}}}}  ;.
\end{tikzpicture}
\caption{A rooted tree on taxa $X = \{ 1,2,3,4,5 \}$, with all non-leaf vertices labelled by their rates.}
\label{f:RateExample}
\end{figure}
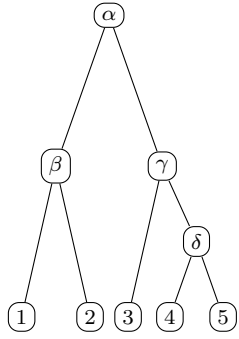

In this space we have

\[ Q_{\alpha} =  \left(  \begin{array}{rrrrr}
-3 & 0 & 1 & 1 & 1 \\
0 & -3 & 1 & 1 & 1 \\
1 & 1 & -2 & 0 & 0 \\
1 & 1 & 0 & -2 & 0 \\
1 & 1 & 0 & 0 & -2 \\
\end{array} \right),
Q_{\beta} =  \left(  \begin{array}{rrrrr}
-1 & 1 & 0 & 0 & 0 \\
1 & -1 & 0 & 0 & 0 \\
0 & 0 & \phantom{-}0 & \phantom{-}0 & \phantom{-}0 \\
0 & 0 & 0 & 0 & 0 \\
0 & 0 & 0 & 0 & 0 \\
\end{array} \right),\]
\[
Q_{\gamma} =  \left(  \begin{array}{rrrrr}
0 & \phantom{-}0 & 0 & 0 & 0 \\
\phantom{-}0 & 0 & 0 & 0 & 0 \\
0 & 0 & -2 & 1 & 1 \\
0 & 0 & 1 & -1 & 0 \\
0 & 0 & 1 & 0 & -1 \\
\end{array} \right),
Q_{\delta} =  \left(  \begin{array}{rrrrr}
0 & 0 & 0 & 0 & 0 \\
0 & 0 & 0 & 0 & 0 \\
\phantom{-}0 & \phantom{-}0 & \phantom{-}0 & 0 & 0 \\
0 & 0 & 0 & -1 & 1 \\
0 & 0 & 0 & 1 & -1 \\
\end{array} \right),\]

\noindent and the matrix algebra is the set

\[ \left\{  \left(  \begin{array}{rrrrr}
* & \beta & \alpha & \alpha & \alpha \\
\beta & * & \alpha & \alpha & \alpha \\
\alpha & \alpha & * & \gamma & \gamma \\
\alpha & \alpha & \gamma & * & \delta \\
\alpha & \alpha & \gamma & \delta & * \\
\end{array} \right): \alpha, \beta, \gamma, \delta \in \mathbb{R}
\right\}  \]

\noindent where $*$ is chosen to give zero row, and column, sum.
\end{example}

\section{The link to graph theory}
\label{graphTsection}

We can also construct the matrix algebra corresponding to a tree $\mathcal{T}$ by considering a certain set of graphs associated with $\mathcal{T}$ that we will refer to as tree-induced graph sets (or TIGS). The basis elements of the matrix algebra will then be the Laplacian matrices of the associated TIGS. 

\begin{definition}
Let $\mathcal{G}_X$ be a set of graphs on vertex set $X$, where $\mathcal{G}_X = \{G_1 = (X,E_1),...,G_\ell = (X,E_\ell)\}$ with edge sets $E_1,...,E_\ell$ disjoint, such that $(X,\cup E_i)$ is the complete graph on $|X|$ vertices. Suppose each graph $G_i \in \mathcal{G}$ is a disjoint union $Z_i \sqcup C_i$ where $Z_i$ is a set of degree-$0$ vertices and $C_i$ is a complete $k$-partite graph for some $k$, and that without loss of generality that $G_1$ contains no degree-$0$ vertices. Finally, given a graph $G_i$ in $\mathcal{G}$, suppose that for each partition $P$ of the $k$ partitions in $C_i$ that contain more than one element, there exists a unique graph $G_j$ where $V(C_j) = V(P)$.  Then we call $\mathcal{G}$ a \textit{tree-induced graph set} (or TIGS). 
\end{definition}

This definition may seem opaque, so we provide an example to aid understanding. While the TIGS have been defined independently of trees, there is a very natural association between TIGS and trees, described in Theorem \ref{t:TIGSvsTrees}. We can therefore refer to a tree and its associated TIGS, with the intention of examining the TIGS using the Laplacian of each graph in the graph set.

\begin{example}
\begin{figure}
\centering
\begin{tikzpicture} 
\SetGraphUnit{1} 
\Vertices{circle}{1,2,3,4,5}
  \path (1) edge              node {} (3)
            edge              node {} (4)
            edge              node {} (5)
        (2) edge              node {} (3)
            edge              node {} (4)
            edge              node {} (5);
\node at (0.35,-2) {$G_\alpha$};
\end{tikzpicture}
\hspace{1cm}
\begin{tikzpicture} 
\SetGraphUnit{1} 
\Vertices{circle}{1,2,3,4,5}
  \path (1) edge              node {} (2);
\node at (0.35,-2) {$G_\beta$};
\end{tikzpicture}
\hspace{1cm}
\begin{tikzpicture} 
\SetGraphUnit{1} 
\Vertices{circle}{1,2,3,4,5}
  \path (3) edge              node {} (5)
            edge              node {} (4);
\node at (0.35,-2) {$G_\gamma$};
\end{tikzpicture}
\hspace{1cm}
\begin{tikzpicture} 
\SetGraphUnit{1} 
\Vertices{circle}{1,2,3,4,5}
  \path (4) edge              node {} (5);
\node at (0.35,-2) {$G_\delta$};
\end{tikzpicture}
\caption{An example of a TIGS. Additionally, these graphs are the $\alpha$-, $\beta$-, $\gamma$- and $\delta$-mrca graphs of the tree in Figure \ref{f:RateExample}, as defined in Definition \ref{d:mrca}.}
\label{f:firstTIGS}
\end{figure}
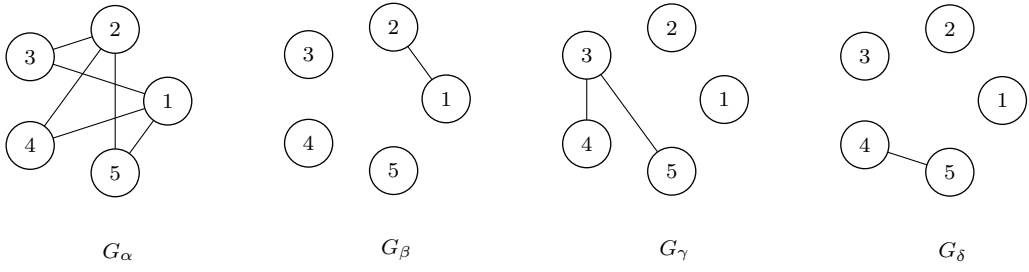

For example, consider the set of graphs depicted in Figure \ref{f:firstTIGS}. We can see that $G_\alpha$ is the only graph in the set that has no degree zero vertices. Further, $G_\alpha$ is a bipartite graph, with partitions $P_1=\{1,2\}$ and $P_2=\{3,4,5\}$. We can then see that $V(G_\beta)$ corresponds to the partition $P_1$, as $C_\beta=P_1$ and $Z_\beta = X \backslash P_1$, and that $C_\beta$ is bipartite with partitions $\{1\}$ and $\{2\}$. Similarly, $G_\gamma$ corresponds to the partition $P_2$ of $G_\alpha$, and $G_\gamma$ is bipartite with partitions $\{3\}$ and $\{4,5\}$. Finally, $G_\delta$ corresponds to the partition $\{4,5\}$ of $G_\gamma$. As the only remaining partitions are singletons, the set $\{G_\alpha,G_\beta,G_\gamma,G_\delta\}$ is a TIGS.
\end{example}

\begin{theorem}
\label{t:TIGSvsTrees}
There exists a bijection between the set of hierarchies on $X$ and the set of tree-induced graph sets on $X$.
\end{theorem}

\begin{proof}
For a cluster $A$ in a hierarchy $H(T)$ with inclusion-maximal subclusters $A_1,...,A_\ell$, we can define the graph $G(A) = (V,E)$ where $V=X$ and $e=(v,w) \in E$ if and only if $v$ and $w$ are in the same inclusion-maximal subcluster $A_i$. This is the disjoint union of the complete graphs $K_{A_i}$. Let $Z$ be the subset of $V$ corresponding to $X \backslash A$. Let $\varphi$ be a function that maps $A$ to $G(A) \cup Z$, and let $\varphi^C$ be the function that maps $A$ to $G^C(A) \cup Z$, where $G^C$ denotes the complement of $G$ (that is, the graph consisting of the same vertex set as $G$ and an edge between vertices $v$ and $w$ if and only if there is not an edge between them in $G$).

Denote by $\phi$ the function that maps $H(T)$ to the set $\{\varphi^C(A) \mid A \in H(T) \}$. This is certainly injective, as $\varphi$ and the operation of taking the complement on the subgraph induced by $G(A)$ are both invertible.  We therefore just need to show that the image of $\phi$ is precisely the set of TIGS.

Suppose we have some TIGS $\mathcal{G} = \{G_1 = (X,E_1),...,G_\ell = (X,E_k)\}$. Let $\mathcal{G}^C = \{C_1^C \cup Z_1,...,C_k^C \cup Z_k \}$, where for $C_i$ the complement is taken on the induced subgraph of $C_i$. Let $H_{i,j}$ be the vertex set of the $j$-th complete graph of $C_i$. We claim that $\mathcal{H} = \{X\} \cup S \cup \{H_{i,j} \mid i \in \{1,...,\ell \}, j \in \{1,...,k\} \}$ forms a hierarchy, where $S$ is the set of singletons on $X$.

Recall that a hierarchy is a set of subsets of $X$ that contains $X$, all singletons and the intersection between two subsets $A$ and $B$ is $A,B$ or empty. Certainly $\mathcal{H}$ contains all singletons, and the intersection of any $H_{i,j}$ with $X$ is $H_{i,j}$, so it only remains to check that for any $H_{i_1,j_1}, H_{i_2,j_2}$ the intersection $H_{i_1,j_1} \cap H_{i_2,j_2}$ is either empty or one of $H_{i_1,j_1}$ or  $H_{i_2,j_2}$.

Suppose $H_{i_1,j_1} \cap H_{i_2,j_2}$ is non-empty. The only way that this is possible is if $V(C_{i_1})$ is a subset of one of the partitions of $C_{i_2}$, or vice versa. But then, respectively, $H_{i_1,j_1} \subseteq H_{i_2,j_2}$ or the reverse, so the intersection $H_{i_1,j_1} \cap H_{i_2,j_2}$ is one of $H_{i_1,j_1}$ or  $H_{i_2,j_2}$.

It follows that $\mathcal{H}$ is a hierarchy and therefore that the stated bijection exists.
\end{proof}

Following the construction in Theorem \ref{t:TIGSvsTrees}, for each interior vertex of a tree, with rate $\alpha$, we can associate a single graph.

\begin{definition}
\label{d:mrca}
Let $\mathcal{T}$ be a tree with associated mrca partition $C_\alpha$. Let $G_\alpha$ be the graph $(V,E)$ where $V=X$ and an edge $e =(x,y) \in E$ if and only if $\omega(mrca(x,y))=\alpha$. Then $G_\alpha(\mathcal{T})$ is referred to as the \textit{$\alpha$-mrca graph of $\mathcal{T}$}.
\end{definition}

Then the set of mrca graphs of $\mathcal{T}$ is the corresponding tree-induced graph set as seen in Theorem \ref{t:TIGSvsTrees}. For example, the corresponding set of mrca graphs of the tree in Figure \ref{f:RateExample} is shown in Figure \ref{f:firstTIGS}.

Recall the folloring standard graph-theoretic definitions.

\begin{definition}
Let $G=(V,E)$ be a graph. Then the \textit{adjacency matrix} $A(G)$ of $G$ is the $|V| \times |V|$ matrix where

\[ (A(G))_{vw} = 
\begin{cases} 
1 & \mbox{ if $(v,w) \in E$}, \\ 
0 & \mbox{otherwise} \\ \end{cases}. \]

The \textit{degree matrix} $D(G)$ of $G$ is the diagonal $|V| \times |V|$ matrix

\[ (D(G))_{vw} = 
\begin{cases} 
deg(v) & \mbox{ if $v=w$}, \\ 
0 & \mbox{otherwise} \\ \end{cases}. \]

Finally, the \textit{Laplacian matrix} $L(G)$ of $G$ is the $|V| \times |V|$ matrix $L(G) = D(G) - A(G)$. We simply write $L,D,A$ if $G$ is clear from context.
\end{definition}

One can then see that the set of negative Laplacians of the associated mrca graphs of $\mathcal{T}$ correspond exactly to the basis elements of the matrix algebra.

\begin{theorem}
For any tree $\mathcal{T}$, interior vertex $u$, and rate $\omega(u) = \alpha$, $Q_\alpha = -L(G_\alpha(\mathcal{T}))$.
\end{theorem}

In the next section we will use the properties of the Laplacians of the associated mrca graphs to prove properties of the resulting matrix algebras.

\section{Algebras induced by trees with distinct rates for each vertex}

We will now show that, for a given tree, the set of rate matrices under matrix multiplication forms a matrix algebra. 

\begin{definition}
A \textit{matrix algebra} is a matrix vector space which is closed under matrix multiplication. A \textit{phylosymmetric algebra} is a matrix set generated from a rooted tree using the previously described method. It always forms an commutative matrix algebra when the rates assigned to the non-leaf vertices are unique (see Theorem \ref{treealgebraThm}). We denote the matrix set generated from a tree $ \mathcal{T} $ by $ \mathcal{Q}_{\mathcal{T}} $.
\end{definition}

In order to prove that the set of rate matrices under matrix multiplication for a given tree $\mathcal{T}$ forms a matrix algebra, it suffices to check that for each possible pair of rate matrices $Q_\alpha, Q_\beta$, the product $Q_\alpha Q_\beta$ is a linear combination of rate matrices derived from $\mathcal{T}$. To do this we will need to be able to refer the relationship between different vertices of $\mathcal{T}$.

\begin{definition}
For a tree $\mathcal{T}$ and two vertices on this tree $u$ and $v$, we say that
\begin{itemize}
    \item $u$ and $v$ are comparable if either $u$ is a descendant of $v$ or the reverse. 
    \item $u$ and $v$ are incomparable if $u$ is neither an ancestor nor a descendant of $v$.
\end{itemize}
\end{definition}

We will also need to refer to different subtrees of $\mathcal{T}$.

\begin{definition}
For a tree $ \mathcal{T} $ which has an internal vertex $u$ with rate $\omega(u) = \alpha$, we define
\begin{itemize}
    \item $T^\alpha$ as the subtree rooted at $u$;
    \item $T_\beta^\alpha$ as the subtree rooted at the child of $u$ that is an ancestor of $v$.
    \end{itemize}
\end{definition}

Finally, we will need to appeal to some classical graph-theoretical results. Theorem \ref{t:k-walks} is folkloric and easily proven (see e.g. \cite{brouwer2011spectra}, Proposition 1.3.1) and Theorem \ref{t:k-walks2} can be proven in an almost identical way. We provide them here as they will be heavily used in the following work.

\begin{theorem}
\label{t:k-walks}
Let $G$ be a graph and $A=A(G)$ its adjacency matrix. Then $(A^k)_{ij}$ is the number of walks of length $k$ on $G$ from vertex $i$ to vertex $j$.
\end{theorem}

\begin{theorem}
\label{t:k-walks2}
Let $G_1=(V,E_1)$ and $G_2=(V,E_2)$ be graphs on the same set of vertices and $A_1=A(G_1),A_2=A(G_2)$ their corresponding adjacency matrices. Consider the multigraph $G^\times = (V,E_1 \cup E_2)$. Then $(A_1A_2)_{ij}$ is the number of walks of length $2$ on $G^\times$ from vertex $i$ to vertex $j$, where the first edge is taken from $E_1$ and the second from $E_2$.
\end{theorem}

We are now in a position to investigate matrix multiplication of elements of $\mathcal{Q}_\mathcal{T}$, by appealing to the structure of the associated TIGS. We will consider squares of a rate matrix first.

\begin{theorem}
Let $u$ be a vertex of a tree $T$ so that $\omega(u)=\alpha$, and let $G_\alpha$ be an $\alpha$-mrca graph, and $Q_\alpha = - L(G_\alpha) = A_\alpha-D_\alpha$ be the $n \times n$ matrix described before.  Suppose $D_\alpha = diag(d_1,...,d_n)$. Then 

\[ (Q_\alpha^2)_{ij} = 
\begin{cases} 
d_i(d_i+1) & \mbox{ if $i=j$}, \\ 
-|T^\alpha| & \mbox{if $i$ and $j$ are in different $k$-partitions of $G_\alpha$} \\
d_i & \mbox{if $i \ne j$ are in the same $k$-partition of $G_\alpha$} \\\end{cases}. \]

Equivalently, if we denote the set of child vertices of $u$ by $C_u$,
\[ Q_\alpha^2 = (1-|T^\alpha|)Q_\alpha + \sum_{\beta \in \omega(C_u)} [(|T^\alpha| - |T^\beta|)(\sum_{\gamma \in \omega(V_u)} Q_\gamma )]. \]  
\label{squareTheorem}
\end{theorem}

\begin{proof}
Since $Q_\alpha = A_\alpha-D_\alpha$, we know $Q_\alpha^2 = A_\alpha^2-D_\alpha A_\alpha - A_\alpha D_\alpha + D_\alpha^2$, and it suffices to consider each of these terms separately. 

As $D_\alpha$ is a diagonal matrix, the last three terms are trivial to calculate. Certainly $D_\alpha^2 = diag(d_1^2,...,d_n^2)$. Further,

\[ (D_\alpha A_\alpha)_{ij} = d_i(A)_{ij} =
\begin{cases} 
0 & \mbox{ if $i,j$ are in the same $k$-partition of $G_\alpha$}, \\ 
d_i & \mbox{otherwise} \\
\end{cases}, \]
and

\[ (A_\alpha D_\alpha)_{ij} = d_j(A)_{ij} =
\begin{cases} 
0 & \mbox{ if $i,j$ are in the same $k$-partition of $G_\alpha$}, \\ 
d_j & \mbox{otherwise.} \\
\end{cases}. \]

Now, by Theorem \ref{t:TIGSvsTrees} we can consider the associated TIGS graph (and in particular $G_\alpha$), and by Theorem \ref{t:k-walks}, $(A_\alpha^2)_{ij}$ is the number of walks of length $2$ from $i$ to $j$ in $G_\alpha$. As $G_\alpha$ is the complete $k$-partite graph for $k$ the number of partitions, if $i,j$ are in the same partition, this is simply the number of vertices of $G_\alpha$ not in this partition, so $d_i$. If they are in different partitions, this is the number of vertices that are in neither the partition containing $i$ nor the one containing $j$. If we denote the partition containing $i$ by $P(i)$ and similarly for $j$, this is $|T^\alpha|-|P(i)|-|P(j)|=d_i+d_j -|T^\alpha|$, since $|P(i)|=|T^\alpha|-d_i$ and $|P(j)| = |T^\alpha|-d_j$. 

To summarise,

\[ (A_\alpha^2)_{ij} = 
\begin{cases} 
d_i & \mbox{if $i,j$ are in the same $k$-partition of $G_\alpha$} \\
d_i+d_j - |T^\alpha| & \mbox{otherwise.} \\\end{cases} \] 

Since $Q_\alpha^2 = A_\alpha^2-D_\alpha A_\alpha - A_\alpha D_\alpha + D_\alpha^2$, we therefore obtain

\[ (Q_\alpha^2)_{ij} = 
\begin{cases} 
d_i(d_i+1) & \mbox{ if $i=j$}, \\ 
-|T^\alpha| & \mbox{if $i$ and $j$ are in different $k$-partitions of $G_\alpha$} \\
d_i & \mbox{if $i \ne j$ are in the same $k$-partition of $G_\alpha$} \\\end{cases}. \]
as required.

Finally, equivalence of the two expressions in the statement of the theorem follows simply by observing the entries of the matrix and applying Remark 1.
\end{proof}

We will now consider multiplication of two rate matrices associated to comparable vertices.

\begin{theorem}
Let $u$ and $v$ be vertices of a tree $T$ so that $\omega(u)=\alpha, \omega(v)=\beta$. Let $G_\alpha, G_\beta$ be $\alpha$- and $\beta$-mrca graphs, and $Q_\alpha = -L(G_\alpha) = A_\alpha-D_\alpha$ and $Q_\beta = -L(G_\beta) = A_\beta - D_\beta$ be the $n \times n$ matrices described before. Finally, suppose that $v$ is a descendant of $u$. Then 

\[ Q_\alpha Q_\beta = (|T_\beta^\alpha|-|T^\alpha|)Q_\beta = Q_\beta Q_\alpha. \]

\label{productTheorem}
\end{theorem}

\begin{proof}
Suppose $D_\alpha = diag(c_1,...,c_n)$ and $D_\beta = diag(d_1,...,d_n)$. Further let $A_\alpha = (a_{ij})$ and $A_\beta = (b_{ij})$.

Since $Q_\alpha Q_\beta = (A_\alpha-D_\alpha)(A_\beta-D_\beta)$, we know $Q_\alpha Q_\beta = A_\alpha A_\beta -A_\beta D_\alpha - A_\alpha D_\beta + D_\alpha D_\beta$, and it suffices to consider each of these terms separately.

We first consider $D_\alpha D_\beta$. As $v$ is a descendant of $u$, any vertex $i$ of $G_\beta$ with non-zero degree is a subset of a single $k$-partition of $G_\alpha$. In particular as $G_\alpha$ is a complete $k$-partite graph $c_i=|T^\alpha| -|T_\beta^\alpha|$ so it follows

\[ (D_\alpha D_\beta)_{ij} = 
\begin{cases} 
(|T^\alpha| -|T_\beta^\alpha|) d_i & \mbox{ if $i=j$ and $i$ is a descendant of $v$}, \\ 
0 & \mbox{otherwise}
\end{cases}. \]

Therefore $(D_\alpha D_\beta) = (|T^\alpha| -|T_\beta^\alpha|)D_\beta$.

We now consider $A_\beta D_\alpha$. Let $(A_\beta)_{ij} = b_{ij}$. As $D_\alpha$ is diagonal, $(A_\beta D_\alpha)_{ij} = b_{ij} c_i$. In particular, $b_{ij}$ is non-zero (in fact $1$) if and only if $i,j$ are both descendants of $v$ and $i$ and $j$ are in different partitions of $G_\beta$. For all such $i,j$, we see $i$ and $j$ are in the same partition of $G_\alpha$, so again $c_i=|T^\alpha| -|T_\beta^\alpha|$. Hence

\[ (A_\beta D_\alpha)_{ij} = 
\begin{cases} 
|T^\alpha| -|T_\beta^\alpha| & \mbox{ if $i,j$ are descendants of $v$ and in separate partitions of $G_\beta$}, \\ 
0 & \mbox{otherwise}
\end{cases}. \]

Therefore $(A_\beta D_\alpha) = (|T^\alpha| -|T_\beta^\alpha|)A_\beta$.

We now consider $A_\alpha D_\beta$. Let $(A_\alpha)_{ij} = a_{ij}$. As $D_\beta$ is diagonal, $(D_\alpha A_\beta)_{ij} = d_j a_{ij}$. In this case, $d_j$ is non-zero if and only if $j$ is a descendant of $v$. But we know all descendants of $v$ are in the same $k$-partition of $G_\alpha$, so it follows that 

\[ (A_\alpha D_\beta)_{ij} = 
\begin{cases} 
d_j & \mbox{ if $j$ is a descendant of $v$ and $i$ is a descendant of $u$ but not $v$}, \\ 
0 & \mbox{otherwise}
\end{cases}. \]

Finally, we consider $A_\alpha A_\beta$. By Theorem \ref{t:TIGSvsTrees} we can consider the associated TIGS graph of $T$ (and in particular $G_\alpha$ and $G_\beta$), and by Theorem \ref{t:k-walks2}, this says that if $G_\alpha = (V,E_1), G_\beta = (V,E_2)$, then by taking the multigraph $G^\times = (V,E_1 \cup E_2)$, $(A_\alpha A_\beta)_{ij}$ is the number of walks of length $2$ on $G^\times$ from vertex $i$ to vertex $j$, where the first edge $e_1$ is taken from $E_1$ and the second edge $e_2$ from $E_2$. We consider $e_2$ first. This is an edge from leaf $k$ in a partition of $G_\beta$ that does not contain $j$ to $j$ itself, of which there are $deg(j) = d_j$ such edges. It follows that, if it exists, $e_1$ is an edge in $G_\alpha$ from the vertex $i$ (which is not a descendant of $v$) to $k$, of which there is only one. Thus

\[ (A_\alpha A_\beta)_{ij} = 
\begin{cases} 
d_j & \mbox{ if $j$ is a descendant of $v$ and $i$ is a descendant of $u$ but not $v$}, \\ 
0 & \mbox{otherwise}
\end{cases}, \]
which means $A_\alpha A_\beta = A_\alpha D_\beta$.

It follows that
\begin{align*}
Q_\alpha Q_\beta & = A_\alpha A_\beta -A_\beta D_\alpha - A_\alpha D_\beta + D_\alpha D_\beta \\
& = D_\alpha D_\beta -A_\beta D_\alpha \\
& = (|T^\alpha| -|T_\beta^\alpha|)D_\beta - (|T^\alpha| -|T_\beta^\alpha|)A_\beta \\
& = (|T^\alpha| -|T_\beta^\alpha|)(D_\beta - A_\beta) \\
& = (|T_\beta^\alpha|-|T^\alpha|)Q_\beta
\end{align*}
as required.

To complete the proof, we see that $Q_\alpha Q_\beta = Q_\beta Q_\alpha$, as $Q_\alpha$ and $Q_\beta$ are symmetric matrices, and their product is a scalar multiple of a symmetric matrix and hence symmetric itself, so by Lemma \ref{commutingLemma} we know that $Q_\alpha$ and $Q_\beta$ commute.
\end{proof}

Finally, we consider multiplication of two rate matrices associated with incomparable vertices.

\begin{theorem}
Suppose that $u$ and $v$ are are incomparable vertices so that $\omega(u)=\alpha$ and $\omega(v) = \beta$. Let $G_\alpha, G_\beta$ be $\alpha$- and $\beta$-mrca graphs, and $Q_\alpha = A_\alpha-D_\alpha$ and $Q_\beta = A_\beta - D_\beta$ be the $n \times n$ matrices described before. Then 

\[ Q_\alpha Q_\beta = 0_{n \times n}. \]
\label{zeroproductTheorem}
\end{theorem}

\begin{proof}
By Theorem \ref{t:TIGSvsTrees} we can consider the associated TIGS graph (and in particular $G_\alpha$ and $G_\beta$), and as $u$ and $v$ are incomparable, $G_\alpha$ and $G_\beta$ can have their vertices partitioned into disjoint sets $A$ and $B$, where $G_\alpha$ only has edges between vertices in $A$, and $G_\beta$ only has edges between vertices in $B$.

It therefore suffices to observe that under an appropriate choice of basis, the Laplacian matrix of each graph is block diagonal, where all non-zero blocks of $Q_\alpha$ correspond to zero blocks of $Q_\beta$, and vice versa. It follows that  

\[ Q_\alpha Q_\beta = 0_{n \times n}. \]
\end{proof}

\begin{theorem}
For a binary phylogenetic tree $ \mathcal{T} $, $ \mathcal{Q}_{\mathcal{T}} $ is an commutative matrix algebra. \label{treealgebraThm}
\end{theorem}

\begin{proof}
We know that $ \mathcal{Q}_{\mathcal{T}} $ is a vector space, closed under matrix products (see Theorems \ref{squareTheorem}, \ref{productTheorem}, \ref{zeroproductTheorem}) and that all matrices in $ \mathcal{Q}_{\mathcal{T}} $ and their products are symmetric, so the space is commutative by Lemma \ref{commutingLemma}.

\end{proof}

\section{Algebras induced by trees with repeated rates}
\label{sectionRepeatedR}

So far we have found that when the rates assigned to tree nodes are unique, the matrix set forms an algebra. Now we explore cases of rates not being unique. We note here that the K2P model is an example of a phylosymmetric algebra with non-unique rates. We see that the tree represented in Figure \ref{k2pExample} gives rise to the K2P model. We know from previous work \citep{fernandez2015lie} that the matrix set for K2P is closed under matrix multiplication. However, in the general case, there is no guarantee that a matrix set will still be closed under matrix multiplication when several rates on the tree are set to be equal. We now explore the conditions that have to be met on such a rooted tree for its rate matrix set to be an algebra.

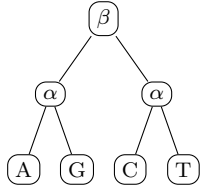
\begin{figure}[h]
\centering
\begin{tikzpicture}[level 1/.style={sibling distance=14mm},level 2/.style={sibling distance=7mm},
  every node/.style = {shape=rectangle, rounded corners,
    draw, align=center}]]
  \node {$ \beta $}
    child[level distance = 1cm]{node{$ \alpha $}
    	child[level distance = 1cm]{node{A}}
    	child[level distance = 1cm]{node{G}} }
    child[level distance = 1cm]{node{$ \alpha $}
    	child[level distance = 1cm]{node{C}}
    	child[level distance = 1cm]{node{T}}}  ;.
\end{tikzpicture}
\caption{A rooted tree on states of DNA with taxa $X = \{$ A, G, C, T $\}$, with all non-leaf vertices labelled by their rates. The phylosymmetric algebra that this tree gives rise to is the K2P model.}
\label{k2pExample}
\end{figure}

\begin{definition}
Let $ \mathcal{T} $ be a tree with at least two non-leaf vertices $u$ and $v$, so that $\omega(u) = \alpha $ and $\omega(v) = \beta $. Let $ \mathcal{T}' $ be a tree with the same topological tree structure and associated rates as $ \mathcal{T} $, with the additional constraint  that $ \alpha = \beta $. (Here, we suppose that there are only two rates on $ \mathcal{T}' $ that are equal) We note that if $ \mathcal{Q}_{\mathcal{T}} = $ span$ \{ Q_{\alpha}, Q_{\beta}, Q_{\gamma}, Q_{\delta}, ... \}_{\mathbb{R}} $ and we define $ Q_{X} = Q_{\alpha} + Q_{\beta} $, then we have $ \mathcal{Q}_{\mathcal{T}'} = $ span$ \{ Q_{X}, Q_{\gamma}, Q_{\delta}, ... \}_{\mathbb{R}} $. If $ \mathcal{Q}_{\mathcal{T}'} $ is a matrix algebra, we say that $ \alpha = \beta $ is a \textit{phylo-algebraic constraint}.
\end{definition}

Labelling two vertices by the same rate is equivalent to adding their rate matrices, so we can consider

\[ (Q_\alpha + Q_\beta)^2 = Q_\alpha^2 + Q_\beta^2 + 2Q_\alpha Q_\beta, \]
as $Q_\alpha Q_\beta = Q_\beta Q_\alpha$ by Lemma \ref{commutingLemma} and Theorem \ref{productTheorem}.

If $u$ is an ancestor of $v$, then by Lemma \ref{squareTheorem} this becomes
\[ Q_\alpha^2 + Q_\beta^2 + 2(|T^\alpha| -|T^\beta|) Q_\beta, \]
and in the particular case that they are incomparable, by Theorem \ref{zeroproductTheorem} we obtain
\[ Q_\alpha^2 + Q_\beta^2. \]

\begin{theorem}
\label{l:siblingID}
If $\mathcal{T}$ is a tree and $u$ and $v$ are siblings so that $\omega(u) = \alpha$ and $\omega(v) = \beta$, and $u$ and $v$ have the same number of leaf descendants, $\alpha = \beta$ is a phylo-algebraic constraint (and hence the resultant matrix algebra is closed).
\end{theorem}

\begin{proof}
Suppose $u$ and $v$ are siblings, and have the same number of leaf descendants (i.e. $|T^\alpha| = |T^\beta|$). Then, by Theorem \ref{squareTheorem}, 

\[ Q_\alpha^2 + Q_\beta^2 = -|T^\alpha|(Q_\alpha + Q_\beta) + \text{scalar multiples of the rate matrices of their descendants,}\]
which is certainly within the generated matrix set. As $u$ and $v$ are siblings, then for any third vertex $w$ with rate $\gamma$, $w$ is an ancestor to both of them, incomparable to both of them, or incomparable to one and a descendant of the other.

If $w$ is an ancestor of both $u$ and $v$, then $(Q_\alpha + Q_\beta)Q_\gamma = (|T^\gamma| -|T^\beta|)(Q_\alpha + Q_\beta)$. If $w$ is incomparable to both, $(Q_\alpha + Q_\beta)Q_\gamma = 0_{n \times n}$. If, $w$ is, say, incomparable to $u$ and a descendant of $v$, then $(Q_\alpha + Q_\beta)Q_\gamma = (|T^\alpha| -|T^\gamma|)Q_\gamma$. This covers all possible cases, as $u$ and $v$ are siblings.

In all three cases the result is clearly in the algebra, so we will always obtain a phylosymmetric algebra.
\end{proof}

\begin{theorem}
\label{l:parentID}
If $\mathcal{T}$ is a tree, $u$ and $v$ are interior vertices such that $\omega(u)=\alpha$ and $\omega(v)=\beta$, and one of $u$ and $v$ is the parent of the other, $\alpha = \beta$ is a phylo-algebraic constraint
\end{theorem}

\begin{proof}
Suppose without loss of generality, $u$ is the parent of $v$. We first consider the tree $T$ without the $\alpha = \beta$ constraint. Using Theorem \ref{t:TIGSvsTrees} we can consider the associated TIGS, in particular $G_\alpha$ and $G_\beta$. Suppose $G_\alpha$ be a complete $k$-partite graph and $G_\beta$ be a complete $k'$-partite graph. In this case we can see that the only change induced to the corresponding TIGS by the $\alpha = \beta$ constraint is that $G_\alpha$ and $G_\beta$ are removed and replaced with $G_\alpha + G_\beta$, where $+$ indicates a graph sum.  Then the resulting mrca graph set is certainly a TIGS, as we can partition $G_\alpha + G_\beta$ into a complete $(k+k'-1)$-partite graph, by applying the $k$-partition of $G_\alpha$ and subpartition the partition consisting of the descendants of $v$ into the $k'$ parts corresponding to $G_\beta$.

The resultant TIGS therefore corresponds to a tree by Theorem \ref{t:TIGSvsTrees}, and therefore by Theorem \ref{treealgebraThm} forms a matrix algebra.
\end{proof}

\begin{observation}
The set of basis matrices obtained in the case of Lemma \ref{l:parentID} coincides exactly with the set of basis matrices of the tree in which the vertices $u$ and $v$ are identified in the graph theoretic sense. Let $T$ be a tree in which there is a union $\cup C_i$ of connected subgraphs of $T$ where each connected subgraph $C_i$ has all rates identified with each other, but not any other connected subgraph $C_j$. Then this will also induce a matrix algebra (indeed a phylosymmetric algebra), as we can sequentially identify parent-child pairs, obtain a matrix algebra corresponding to a tree and then identify another parent-child pair.
\end{observation}

\begin{theorem}
Let $\mathcal{T}$ be a tree with unique rates and $ \mathcal{Q}_{\mathcal{T}} $ be the phylosymmetric algebra of $ \mathcal{T} $. If $u$ and $v$ are interior vertices so that $\omega(u) = \alpha$ and $\omega(v) = \beta$, we define $\mathcal{Q}_{\mathcal{T}}^{\alpha=\beta}$ as the matrix set generated from setting $\alpha = \beta$. $\mathcal{Q}_{\mathcal{T}}^{\alpha=\beta}$ is a matrix algebra if and only if one of the following is true:
\begin{enumerate}
    \item $u$ is a parent of $v$ or vice versa;
    \item $u$ and $v$ are siblings and have the same number of leaf descendants.
\end{enumerate}
\end{theorem}

\begin{proof}

For an added constraint $\alpha = \beta$, we let $ Q_{X} = Q_{\alpha} + Q_{\beta} $. We can show that $ \mathcal{Q}_{\mathcal{T}}^{\alpha = \beta} $ is not a matrix algebra by showing that products in the space cannot be written as linear combinations that include $Q_{X}$ but do not include $Q_{\alpha}$ and $Q_{\beta}$.

First we assume that $ \mathcal{Q}_{\mathcal{T}}^{\alpha=\beta} $ is a matrix algebra. There are five possible ways to describe the positions of two vertices $u$ and $v$ on a tree: 
\begin{enumerate}
    \item There exists a vertex $w$ such that $w$ is a descendant of $u$ and an ancestor of $v$.
    \item There exists a vertex $w$ such that $u$ and $w$ are incomparable and $v$ is a descendant of $w$.
    \item There exists a vertex $w$ with rate $\gamma$ such that $u$ and $v$ are child vertices of $w$ and $ |T^{\alpha}| \not= |T^{\beta}| $.    
    \item There exists a vertex $w$ with rate $\gamma$ such that $u$ and $v$ are child vertices of $w$ and $ |T^{\alpha}| = |T^{\beta}| $.
    \item The vertex $u$ is a parent of $v$ or vice versa.
\end{enumerate}

In Case 1 we see that
\begin{align*}
    Q_{\gamma}Q_{X} =& Q_{\gamma}(Q_{\alpha} + Q_{\beta}) \\
    =& Q_{\gamma}Q_{\alpha} + Q_{\gamma}Q_{\beta} \\
    =& -n_{1}Q_{\gamma}-n_{2}Q_{\beta} (\because \text{Theorem \ref{productTheorem} } \text{ where } n_{i} \in \mathbb{N}),
\end{align*}

\noindent as $n_{1} \not= n_{2}$, therefore $ \alpha = \beta $ is not a phylo-algebraic constraint and $ \mathcal{Q}^{\alpha = \beta}_{\mathcal{T}} $ is not a matrix algebra.

For Case 2 we let $u$ and $w$ be incomparable and $v$ be a descendant of $w$. We then have

\begin{align*}
    Q_{\gamma}Q_{X} =& Q_{\gamma}(Q_{\alpha} + Q_{\beta}) \\
    =& Q_{\gamma}Q_{\alpha} + Q_{\gamma}Q_{\beta} \\
    =& (|T^\gamma| - |T^\beta|)Q_\beta.
\end{align*}

As this set of matrices are linearly independent, any scalar multiple of $Q_\beta$ is not able to be generated by the set, and so this product is not contained within the space.

In Case 3, if we denote the set of child vertices of $w$ by $C_w$,
\begin{align*}
    Q_{\gamma}^2 &= (1-|T^\gamma|)Q_{\gamma} + \sum_{\delta \in \omega(C_w)} [(|T^\gamma| - |T^\delta|)(\sum_{\epsilon \in \omega(V_w)} Q_\epsilon )] \\
    &= (|T^\gamma| - |T^\alpha|)Q_{\alpha} + (|T^\gamma| - |T^\beta|)Q_{\beta} + \text{ other matrix terms linearly independent of } Q_{\alpha}\text{ and }Q_{\beta}.\\
\end{align*}

\noindent As we know that $ |T^{\alpha}| \not= |T^{\beta}| $, we can see that under these circumstances, $ \mathcal{Q}_{\mathcal{T}}^{\alpha = \beta} $ is not a matrix algebra.

So we see that only cases 4 and 5 remain, and both produce matrix algebras by Lemmas \ref{l:parentID} and \ref{l:siblingID} respectively.

The theorem follows.
\end{proof}

\section{Discussion}

In Section \ref{backgroundSection}, we introduced a set of matrices associated with trees that had rates associated to each non-leaf vertex. In Section \ref{graphTsection}, we derived results on the multiplication of these matrices, and showed, in the case that each rate is unique, that the matrices form a matrix algebra. In Section \ref{sectionRepeatedR}, we extended this result to completely characterise all conditions for which the matrices form a matrix algebra when two rates are identical, and derived sufficient conditions for simple cases of arbitrarily many equal rates.

In previous work it has been found that building phylogenetic models with a focus on mathematical, rather than biological, properties can produce models which are computationally faster to use and can address biological problems that had not previously been considered \citep{sumner2012lie,sumner2017,shore2015lie}. Development of phylogenetic models also presents new applications of, and new problems in, linear algebra, graph theory and other areas of mathematics \citep{steel2016phylogeny}. Phylosymmetric algebras are an application of both linear algebra and graph theory in phylogenetics which has previously been unexplored. We hope that future research in this area will provide similarly valuable results. In particular, future work could characterise \textit{all} conditions for which a tree with a given set of associated rates form a matrix algebra. In addition, a characterisation of which matrix algebras are induced by trees would also be interesting and may lead to a better structural understanding of rooted trees.

Another avenue of possible research from this point is development of phylogenetic models. We have shown that phylosymmetric algebras have desirable mathematical properties. \citet{sumner2012lie,shore2015lie} have shown that such mathematical properties are desirable in rate substitution models. To use these algebras for rate substitution models in DNA would not provide much in the way of new ground given the broad literature of DNA rate substitution models (for example, \citet{fernandez2015lie} for example provides a list of all parameterised DNA models with purine/pyrimidine symmetry which are closed under multiplication). Although, as discussed in Section \ref{sectionRepeatedR}, we note that the K2P model is an example of a phylosymmetric algebra.

In amino acid substitution models, however, empirical models are most commonly used (\citet{le2008improved} for example) with very few parameterised models having been developed an utilised. The current parameterised amino acid substitution models \citep{yang1998models, adachi1996model} have between 24 and 190 parameters and are not constructed with desirable mathematical properties. To fill this gap, our method of rate matrix construction could be used to build a suite of parameterised amino acid substitution matrices with between 3 and 19 parameters. Having a smaller number of parameters makes computations faster (and hence more computational power can be dedicated to checking the robustness of results) \citep{mello2016fast} and makes the process of interpreting the fitted parameters a much simpler task. 

This proposed method of amino acid substitution matrix generation is distinct from all existing amino acid substitution matrices as our proposed approach features a set of parameterised matrices with a low number of parameters. These models have desirable mathematical properties and, given we can build the initial trees with splits that represent characteristics of amino acids such as polarity, the parameters convey biological significance. As well as such models being mathematically tractable, they have also already been shown to have real biological applications and correlate with biological data as shown by \citet{shore2019}.

\bibliographystyle{spbasic}
\bibliography{template}

\end{document}